\documentclass[conference,letterpaper,twoside]{IEEEtran}

\usepackage[cmex10]{amsmath}
\usepackage{graphicx,epic,eepic,epsfig,latexsym,amssymb,verbatim,color,xcolor,revsymb}

\usepackage[abs]{overpic}

\usepackage[english]{babel}
\usepackage[T1]{fontenc}
\usepackage{mathrsfs}
\usepackage[colorlinks]{hyperref}
\usepackage{enumerate}
\usepackage{graphicx}
\usepackage{mathtools}
\usepackage[export]{adjustbox}
\usepackage{microtype}
\usepackage{MnSymbol}
\usepackage{lipsum}

\usepackage{colortbl}

\usepackage{tikz}
\usepackage{tikzscale}
\usepackage{pgfplots}
\usetikzlibrary{calc}
\usetikzlibrary{patterns}
\usetikzlibrary{decorations.pathreplacing}
\usepackage{blox}

\bXLineStyle{-}

\definecolor{sccolor}{RGB}{62, 150, 81}
\definecolor{spcolor}{RGB}{57, 106, 177}

\definecolor{dullmagenta}{rgb}{0.4,0,0.4}   
\definecolor{darkblue}{rgb}{0,0,0.4}

\usepackage{hyperref}
\hypersetup{
	unicode=false,          
	pdftoolbar=true,        
	pdfmenubar=true,        
	pdffitwindow=false,     
	pdfstartview={FitH},    
	pdftitle={My title},    
	pdfauthor={Author},     
	pdfsubject={Subject},   
	pdfcreator={Creator},   
	pdfproducer={Producer}, 
	pdfkeywords={keyword1} {key2} {key3}, 
	pdfnewwindow=true,      
	linktocpage=true,
	colorlinks=true,        
	linkcolor=red,          
	citecolor=blue,  
	filecolor=Magenta,      
	urlcolor=blue,    
}

\usepackage{doi}
\usepackage{caption,subcaption}
\usepackage{enumitem}
\usepackage{shadethm}
\usepackage{epigraph}

\makeatletter
\newcommand{\opnorm}{\@ifstar\@opnorms\@opnorm}
\newcommand{\@opnorms}[1]{%
	\left|\mkern-1.5mu\left|\mkern-1.5mu\left|
	#1
	\right|\mkern-1.5mu\right|\mkern-1.5mu\right|
}
\newcommand{\@opnorm}[2][]{%
	\mathopen{#1|\mkern-1.5mu#1|\mkern-1.5mu#1|}
	#2
	\mathclose{#1|\mkern-1.5mu#1|\mkern-1.5mu#1|}
}
\makeatother

\makeatletter
\ifx\@NODS\undefined%

\else%
\fi
\makeatother

\makeatletter
\ifx\@NODS\undefined%

\let\mathbb=\mathds
\else%
\fi
\makeatother

\usepackage{xargs}
\usepackage[prependcaption]{todonotes}
\newcommandx{\eric}[2][1=]{\todo[inline, author={Eric}, linecolor=yellow,backgroundcolor=yellow!25,bordercolor=yellow,#1]{#2}}

\newcommandx{\ericside}[2][1=]{\todo[author={Eric}, linecolor=yellow,backgroundcolor=yellow!25,bordercolor=yellow,#1]{#2}}

\DeclareMathOperator{\tr}{Tr}

\DeclareMathOperator{\cN}{\mathcal{N}}
\DeclareMathOperator{\cU}{\mathcal{U}}

\newcommand{\hrightarrow}{\rightarrow}

\newcommand\restr[2]{{
  \left.\kern-\nulldelimiterspace 
  #1 
  \vphantom{\big|} 
  \right|_{#2} 
  }}
\newcommand{\wsigma}{\widetilde{\sigma}}

\newcommand{\bra}[1]{\langle #1 |}
\newcommand{\ket}[1]{| #1 \rangle}

\newcommand{\be}{{\mathbf e}}

\newcommand{\cT}{{\mathcal{T}}}

\newcommand{\cE}{{\mathcal{E}}}

\newcommand{\cD}{{\mathcal{D}}}

\newcommand{\id}{{\mathrm{id}}}

\newcommand{\cV}{{\mathcal{V}}}

\def\reff#1{(\ref{#1})}

\def\0{{\mathbf{0}}}
\def\1{{\mathbf{1}}}
\def\2{{\mathbf{2}}}
\def\3{{\mathbf{3}}}
\def\4{{\mathbf{4}}}
\def\5{{\mathbf{5}}}
\def\6{{\mathbf{6}}}

\def\7{{\mathbf{7}}}
\def\8{{\mathbf{8}}}
\def\9{{\mathbf{9}}}

\def\be{\begin{equation}}
\def\ee{\end{equation}}
\def\bea{\begin{eqnarray}}
\def\eea{\end{eqnarray}}
\def\reff#1{(\ref{#1})}

\newtheorem{theorem}{Theorem}

\newtheorem{proposition}[theorem]{Proposition}

\newtheorem{remark}[theorem]{Remark}  

\newcommand\qedsymbol{$\blacksquare$}
\newcommand\qed{\hfill\qedsymbol}

\newlength{\blank}
\newenvironment{proof-of}[1][{\hspace{-\blank}}]{{\medskip\noindent\textit{Proof~{#1}.\ }}}{\hfill\qedsymbol}
\newenvironment{proof}{{\medskip\noindent\textit{Proof.\ }}}{\hfill\qedsymbol}

\newcommand{\oR}{{\overline{R}}}
\newcommand{\oP}{{\overline{P}}}

\begin{document}

\title{Convexity and Operational Interpretation of the \protect\\
       Quantum Information Bottleneck Function}

\author{%
  \IEEEauthorblockN{Nilanjana Datta\IEEEauthorrefmark{1}}
  \IEEEauthorblockA{\IEEEauthorrefmark{1}%
                    DAMTP\\
                    Centre of Mathematical Sciences\\
                    University of Cambridge\\
                    Cambridge CB3 0WA, UK\\
                    Email: n.datta@statslab.cam.ac.uk}
  \and
  \IEEEauthorblockN{Christoph Hirche\IEEEauthorrefmark{2}}
  \IEEEauthorblockA{\IEEEauthorrefmark{2}%
                    Grup d'Informaci\'{o} Qu\`{a}ntica\\
                    Departament de F\'{\i}sica\\
                    Universitat Aut\`{o}noma de Barcelona\\
                    08193 Bellaterra (Barcelona), Spain\\
                    Email: christoph.hirche@uab.cat}
  \and
  \IEEEauthorblockN{Andreas Winter\IEEEauthorrefmark{2}\IEEEauthorrefmark{3}}
  \IEEEauthorblockA{\IEEEauthorrefmark{3}%
                    ICREA---Instituci\'o Catalana de\\
                    la Recerca i d'Estudis Avan\c{c}ats\\
                    Pg.~Lluis Companys, 23\\
                    08010 Barcelona, Spain\\
                    Email: andreas.winter@uab.cat}
}
 
\date{14 March 2019}
	
\maketitle 

\begin{abstract}
In classical information theory, the information bottleneck method (IBM) can be 
regarded as a method of lossy data compression which focusses on preserving 
meaningful (or relevant) information. As such it has of late gained a lot of 
attention, primarily for its applications in machine learning and neural networks. 
A quantum analogue of the IBM has recently been defined, and an attempt at 
providing an operational interpretation of the so-called quantum IB function 
as an optimal rate of an information-theoretic task, has recently been made 
by Salek \emph{et al.} 
The interpretation given by these authors is however incomplete, 
as its proof is based on the conjecture that the quantum IB function 
is convex. Our first contribution is the proof of this conjecture.

Secondly, the expression for the rate function involves certain 
entropic quantities which occur explicitly in the very definition of the 
underlying information-theoretic task, thus making the latter somewhat contrived. 
We overcome this drawback by pointing out an alternative operational interpretation 
of it as the optimal rate of a \emph{bona fide} information-theoretic task, namely 
that of quantum source coding with quantum side information at the decoder, 
which has recently been solved by Hsieh and Watanabe. 
We show that the quantum IB function characterizes the rate region of this task,

We similarly show that the related \emph{privacy funnel} function is concave 
(both in the classical and quantum case).  However, we comment that it is 
unlikely that the quantum privacy funnel function can characterize the 
optimal asymptotic rate of an information theoretic task, since even its 
classical version lacks a certain essential additivity property.
\end{abstract}



\section{Introduction} 
\label{sec:introduction}
Consider a given pair of random variables $(X,Y)$ with joint probability distribution
$p_{XY}$. In this paper, all random variables are considered to be discrete, taking 
values $x,y$ in finite alphabets ${\mathcal{X}}$ and ${\mathcal{Y}}$, respectively. 
Tishby \emph{et al.} \cite{tishbyIBM} introduced the notion of the meaningful or 
{\em{relevant information}} that $X$ provides about $Y$. They formalized this notion
as a constrained optimization problem of finding the optimal compression of $X$ 
(to a random variable $W$, say) which 
still retains maximum information about $Y$. The authors of \cite{tishbyIBM} named this problem {\em{Information Bottleneck}} since $W$ can be viewed as the result of squeezing the information that $X$ provides about $Y$ through a ``bottleneck''. The information bottleneck can be regarded as a problem of lossy data compression for a source defined by the random variable $X$, {{in presence of side information given by $Y$}}. The standard theory of lossy data compression introduced by Shannon \cite{shannon59} is {\em{rate distortion theory}}, which deals with the trade-off between the rate of lossy compression and the average distortion of the distorted signal (see also~\cite{berger71,cover}). The Information Bottleneck Method (IBM) can be considered as a generalization of this theory, in which the distortion measure between $X$ and $W$ is determined by the joint distribution $p_{XY}$. This method has found numerous applications, \emph{e.g.}~in investigating deep neural networks~\cite{tishby2015deep, shwartz2017opening}, video processing~\cite{hsu2006video}, clustering~\cite{slonim2000doc} and polar coding~\cite{StarkPolar}.

{{The constraint in the above-mentioned optimization problem, is given as a lower bound, say $I_Y$, on the mutual information, $I(Y;W) = H(Y) + H(W) - H(YW)$, since the latter is a measure of 
the information about $Y$ contained in $W$. Here $H(Y)$ denotes the Shannon entropy of $Y$, \emph{i.e.}~if $Y$ has a probability mass function $\{p(y)\}_{y \in {\mathcal Y}}$, where ${{\mathcal Y}}$ is a finite alphabet, then $H(Y)= - \sum_{y \in {\mathcal Y}}p(y) \log p(y)$. The {\em{rate function}} of the IBM, the so-called {\em{IB function}}, is a function of this bound and is given by 
\begin{align}\label{cl-ra-min}
R(I_Y) &= \min_{p(w|x)\atop{I(Y;W) \geq I_Y}} I(X';W), \quad {\hbox{for}} \ I_Y \geq 0,
\end{align}
where $X'=X$, and the minimization is over the set of conditional probabilities $\{p(w|x)\}$, with $w$ denoting values taken by the random variable $W$.

A dual quantity, which gives an expression for the information $I_Y$ 
as a function of the rate $R$, is given by
\begin{align}
  \label{cl-ra-max}
  I_Y(R) &= \max_{p(w|x)\atop{I(X';W) \leq R}} I(Y;W), \quad {\hbox{for}} \ R \geq 0.
\end{align}
It was shown in~\cite[Lemma~10]{GNT03} that the optimization problems in (\ref{cl-ra-min}) 
and (\ref{cl-ra-max}) are indeed dual to each other, meaning that $R(I_Y)$ and $I_Y(R)$ 
are equivalent quantities, in the sense that they define the same curve with 
switched axes for $0\leq I_Y \leq I(W;Y)$ and $0\leq R \leq R(I(W;Y))$.}
In other words, $R$ and $I_Y$ are functions inverse to each other.

As a matter of fact, in~\cite{WW1975}, the following closely related 
optimization problem was investigated:
\begin{align}
  \label{Eq:condEntrOpt}
  F(a) = \min_{\substack{ p(w|x) \\ H(X|W)\geq a}} H(Y|W),
\end{align}
where $H(Y|W):= H(YW) - H(W)$ is the conditional entropy.
It was furthermore shown that $F(a)$ is always convex. 
It can easily be seen that $I_Y(a) = H(Y) - F(H(X)-a)$;
it follows that $I_Y(R)$ is concave and $R(I_Y)$ is convex. 

\subsection*{Operational interpretation of the classical IB function}

An operational interpretation of the IB function is obtained by 
considering the function $\widehat F(a) = H(Y) - I_Y(a)$.
More precisely, its interpretation follows from that of $\widehat F(a)$ via 
the so-called Wyner-Ahlswede-K\"orner (WAK) problem~\cite{Wyner75, AK75}.  
The setting considered for the WAK problem is that of source coding 
with side information at the decoder. In this paper we will be 
concerned with a generalization of this task to the quantum setting,
so let us review the WAK problem briefly.

\begin{figure}[ht]
  \includegraphics[width=0.5\textwidth]{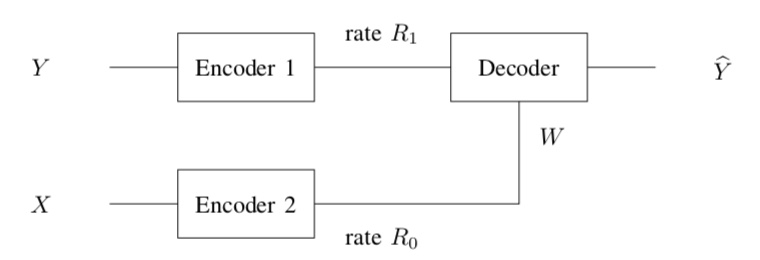}
  \caption{\small Schematic of the Wyner-Ahlswede-K\"orner problem.}
  \label{fig:WAK}
\end{figure}
    
The WAK problem concerns encoding information about one random variable such that it can be reconstructed using information about another (correlated) random variable. Let $X$ and $Y$ be two correlated random variables. One encodes $X$ and $Y$ separately, at rates $R_0$ and $R_1$, respectively. Both encodings are available to the decoder. A pair of rates $(R_0,R_1)$ is called achievable if it allows for exact reconstruction of $Y$ in the asymptotic i.i.d.~setting. Since we do not aim to recover $X$, its encoding is considered as side information at the decoder provided by a \textit{helper}. It was found independently in~\cite{Wyner75} and~\cite{AK75} that the minimal achievable rate $R_1$ under the constraint $R_0\leq a$ is given by $\widehat F(a)$.

\section{Quantum information bottleneck}
\label{sec:QIBM}
A quantum generalization of the information bottleneck was first proposed by 
Grimsmo and Still~\cite{Grimsmo16}. They considered the following problem: 
let $\rho_X$ denote the state of a quantum system $X$, and let $\psi_{XR}$ 
denote its purification. The purifying reference system, $R$, is sent 
through a quantum channel, \emph{i.e.}~a linear completely positive trace-preserving 
(CPTP) map ${\mathcal R}: R \to Y$. It is only the information in $Y$ 
which is deemed important or {\em{relevant}}. The aim is to 
find an optimal encoding (\emph{i.e.}~compression) ${\mathcal N}: X \to M$ of $X$ 
into a quantum ``memory'' system $M$, which enables the retention of as 
much information about $Y$ as possible, without storing any unnecessary data. 
They quantified the information encoded about the initial data via the 
quantum mutual information $I(M;R)$, and the information available in 
$M$ about $Y$ by the quantum mutual information $I(M;Y)$. The optimal 
encoding is then the solution over an optimization problem in which $I(M;Y)$ 
is maximized over all possible channels, such that $I(M;R)$ is below a given threshold. 

Formally, for a given bipartite quantum state $\rho_{XY}$, 
the {\em{quantum IB function}}, is defined
through the following constrained optimization problem:
\begin{align}
  \label{q-IBM}
  R_q(a) &= \inf_{\cN^{X\to W}\atop{I(Y;W)_\sigma \geq a}}
                    I(X';W)_{\tilde{\tau}}, \quad \text{for} \ a\geq 0,
\end{align}
where the optimization is over all linear CPTP maps $\cN^{X\to W}$ mapping 
states of $X$ to states of $W$, under the given constraint. In the above, 
${\tilde{\tau}}_{X'W} := ({\rm{id}}_{X'} \otimes \cN^{X\to W})\tau_{X'X}$, 
where $\tau_{X'X}$ is a purification of $\rho_X$, and 
$\sigma_{WY} := ( \cN^{X\to W} \otimes {\rm{id}}_{Y} )\rho_{XY}$. 
Hence $\rho_X = \tr_{X'}\tau_{X'X} = \tr_Y \rho_{XY}$.
Here, $I(Y;W)_\sigma:= S(\sigma_Y) + S(\sigma_W) - S(\sigma_{YW})$ denotes the 
quantum mutual information, with $S(\sigma_Y) := - \tr (\sigma_Y \log \sigma_Y)$ 
being the von Neumann entropy.

However, the operational significance of this task remained unclear.
Later, Salek \emph{et al.}~\cite{salek2017} attempted to give an operational 
interpretation to the quantum IB function.
They showed that it is the
optimal asymptotic rate of a certain information-theoretic task, {\em{under the assumption}}
that the quantum IB function is convex. The task that they considered was the following {\em{constrained
version}} of entanglement-assisted lossy data compression, in the communication paradigm, with a suitable choice of distortion measure. The state ($\rho_X$) to be compressed is in the possession of the sender (say, Alice), and  
is the reduced state of a bipartite state $\rho_{XY}$. Alice does not have
access to the system $Y$. There is a noiseless classical channel between the her and the receiver (say, Bob). Alice and Bob also have prior shared entanglement. The {\em{relevant information}} that the state of the quantum system $X$ provides about that of $Y$ is quantified by the quantum mutual information
$I(X;Y)_\rho$. Alice compresses $\rho_X$ and sends it through the noiseless classical channel to Bob, who then decompresses the data. Alice and Bob each use 
their share of entanglement in their respective compression and decompression tasks. The aim of the task is to find the optimal rate (in bits) of data compression under the constraint that the relevant information does not drop below a certain pre-assigned threshold.

We will complete Salek \emph{et al.}'s work by showing in the following
section that the quantum IB function is indeed convex, as they had conjectured.

At the same time, one might argue that the information-theoretic task considered 
is somewhat contrived, since it includes a constraint on an entropic function, 
namely a quantum mutual information, in its definition. Usually, the definition of 
an information-theoretic task is entirely operational, and the entropic 
quantities characterizing the optimal rates arise solely as a result of 
the computation. We will address this criticism by providing such an
interpretation in section \ref{sec:op}.

\section{Convexity of the QIB function}
\label{sec:convexity}
Our first result is the proof of the convexity of the quantum IB function, 
as conjectured in~\cite{salek2017}. To do so, we start with the observation that
the quantity $R_q(a)$ of Eq.~\eqref{q-IBM} can be expressed equivalently 
as follows:
\begin{align}
  \label{rq-min}
  R_q(a) &= \inf_{\cN^{X\to W}\atop{I(Y;W)_\sigma \geq a}} I(YR;W)_\sigma, \quad \text{for}\ a \geq 0. 
\end{align}
To see why, let $\psi_{XYR}$ be a purification of $\rho_{XY}$. 
Since it is also a purification of $\rho_X$, there must exist an 
isometry $\cV: X' \to YR$ such that
\begin{align}
  (\id_{X} \otimes \cV^{X' \to YR}) \tau_{XX'} = \psi_{XYR}.
\end{align}
Then, defining $\sigma_{WYR} := (\cN^{X \to W} \otimes \id_{YR}) \psi_{XYR}$, 
by the invariance of the mutual information under isometries, and the 
fact that ${\tilde{\tau}}_{W}:= \tr_{X'} {\tilde{\tau}}_{X'W} = \sigma_W$, 
with ${\tilde{\tau}}_{X'W}$ defined as in Eq.~(\ref{q-IBM}), we have  
\begin{align}
I(X';W)_{\tilde{\tau}} = I(YR;W)_\sigma.
\end{align}

The representation Eq. (\ref{rq-min}) has the benefit of referring to information quantities of the
same tripartite state (rather than two different ones), both in the 
objective function and the optimization constraint.

\begin{theorem}
\label{thm1}
The quantum IB function $R_q(a)$ defined through Eq. (\ref{q-IBM}), is convex, \emph{i.e.}
\begin{align}
  R_q(\lambda a_0 + (1-\lambda) a_1) & \leq \lambda R_q(a_0) + (1-\lambda) R_q(a_1),
\end{align}
for all $\lambda \in [0,1]$ and $a_0,\, a_1 \geq 0$.
\end{theorem}

\begin{proof}
Let $\cN_0^{X \to W}$ and $\cN_1^{X \to W}$ be the optimizing channels in Eq. (\ref{rq-min})
for $a_0$ and $a_1$ respectively, such that 
$I(Y;W)_{\sigma_0} \geq a_0$ and  $I(Y;W)_{\sigma_1} \geq a_1$,
where for $i=0,1$,
\begin{align}
  \sigma_i &= (\cN_i^{X \to W} \otimes \id_{YR}) \psi_{XYR}. 
\end{align}
Hence $R_q(a_i) = I(YR;W)_{\sigma_i}$ for $i=0,1$.
Next consider the flagged channel $\cN^{X\to WW'}$, 
with a qubit $W'$, defined as follows: for $\lambda \in [0,1]$, let
\begin{align}
  \cN^{X\to WW'} &:= \lambda \cN_0 \otimes |0\rangle\!\langle 0|_{W'} 
                      + \lambda \cN_1 \otimes |1\rangle\!\langle 1|_{W'},
\end{align}
Then $\sigma_{YRWW'} := (\cN^{X\to WW'} \otimes \id_{YR})\psi_{XYR}$ is 
a block-diagonal state with diagonal blocks $\lambda \sigma_0^{WYR}$ and
$(1-\lambda) \sigma_1^{WYR}$ respectively. 
Thus,
\begin{align}
  I(Y;WW')_\sigma &=    \lambda I(Y;W)_{\sigma_0} + (1-\lambda) I(Y;W)_{\sigma_1}\nonumber\\
                  &\geq \lambda a_0 + (1-\lambda) a_1, \text{ and}                        \\
  I(YR;WW')_\sigma &= \lambda I(YR;W)_{\sigma_0} + (1-\lambda) I(YR;W)_{\sigma_1}\nonumber\\
                   &= \lambda R_q(a_0) + (1-\lambda) R_q(a_1).
\end{align} 
Therefore, 
\begin{align}
R_q(\lambda &a_0 + (1-\lambda) a_1) 
         =    \inf_{\cN^{X \to W}\atop{I(Y;W) \geq \lambda a_0 + (1-\lambda) a_1}} I(YR;W)\nonumber\\
        &\leq I(YR;WW')_\sigma
         =    \lambda R_q(a_0) + (1-\lambda) R_q(a_1),
\end{align}
concluding the proof.
\end{proof}

\medskip
This not only serves to complete the proof of the operational interpretation of 
the quantum IB function given in~\cite{salek2017}, but is also of independent interest.

\section{Operational interpretation}
\label{sec:op}
We now show that the quantum IB function precisely 
characterizes the achievable rate region of a \emph{bona fide} information 
theoretic task, namely, that of \emph{quantum source coding with quantum 
side information at the decoder} \cite{HsiehWatanabe},
described below and summarized in Theorem~\ref{thmOptRate}.


\subsection*{The task: quantum Wyner-Ahlswede-K\"orner problem}

Let us start by giving an explicit description of the task, which is
a quantum version of the WAK problem, following the work of Hsieh
and Watanabe~\cite{HsiehWatanabe}. It involves three 
parties Bob -- the sender (or encoder), Charlie -- the receiver (or decoder), 
and Alice -- the helper. In contrast to the classical setting, one furthermore 
allows for prior shared entanglement between the helper and the decoder.

Suppose a source provides Alice {(the helper)} and Bob {(the encoder)} 
with the $X$ and $Y$ parts of a quantum state 
$\psi_{X^nY^nR^n} = \left(\psi_{XYR}\right)^{\otimes n}$, respectively,  
with $R$ denoting an inaccessible, purifying reference system. 
Suppose, moreover, that Alice shares entanglement, given by 
the state $\Phi_{T_XT_C}$ with a third party, Charlie (the decoder), 
to whom she can send qubits via a system $C$, at a rate
$Q_X = \frac1n {\log |C|}$. 
Bob, on the other hand, can send qubits to Charlie, via a system ${\widetilde{C}}$, 
at a rate $Q_Y = \frac1n \log |\widetilde{C}|$. Collaborating together, 
Charlie's task is to decode $C$, $T_C$ and $\widetilde{C}$ to a high-fidelity 
approximation of $CT_CY^n$ in the asymptotic limit ($n\to \infty$). The 
encoding maps used by Alice and Bob, and the decoding map used by Charlie, 
are all linear CPTP maps.

\begin{figure}[ht]
  \includegraphics[width=0.5\textwidth]{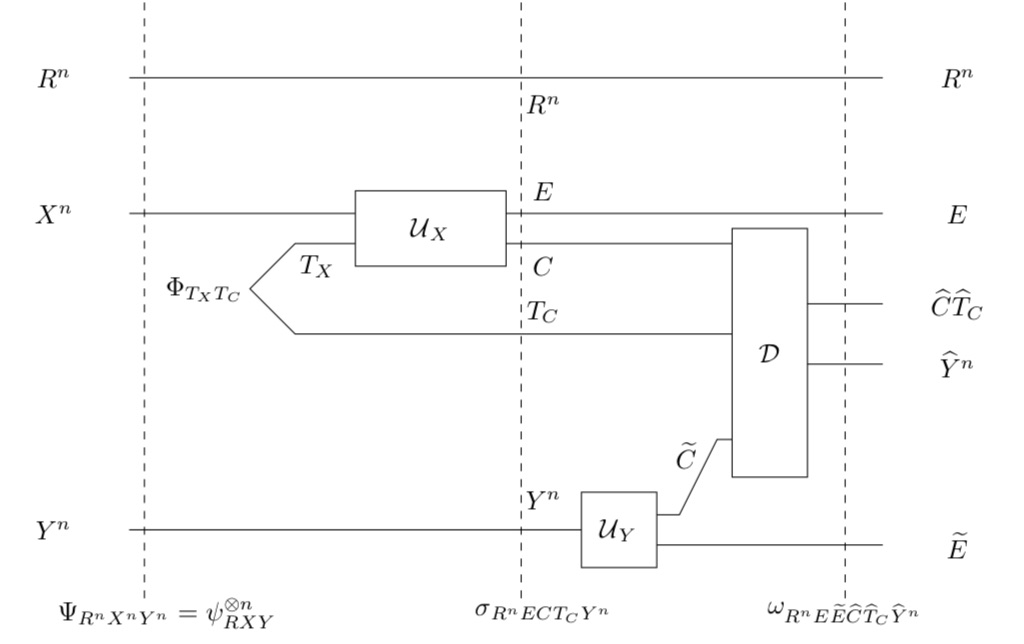}
  \caption{\small Schematic of the quantum source coding with side information task,
           the quantum version of the WAK problem.}
  \label{fig:quantum}
\end{figure}

In Fig.~\ref{fig:quantum}, we show a circuit diagram of the most general
protocol. 
Alice's encoding map $\cE_X:X^n T_X \to C$ has a Stinespring isometry 
$\cU_{X}: X^n T_X \hrightarrow CE$. Similarly, Bob's encoding map
$\cE_Y:Y^n \to \widetilde{C}$ has a Stinespring isometry 
$\cU_{Y}: Y^n \hrightarrow \widetilde{C}\widetilde{E}$. Finally, 
we denote Charlie's decoding map as 
$\cD:C T_C \widetilde{C} \to \widehat{C}\widehat{T}_C\widehat{Y}^n$; 
see Fig.~\ref{fig:quantum}.
The objective is to ensure that the fidelity of the protocol satisfies
\begin{align}\label{fid}
  F_n &:= \tr \left(\sigma_{R^n E C T_C Y^n}\omega_{R^nE\widehat{C}\widehat{T}_C\widehat{Y}^n}\right)
      \to 1 \text{ as } n \to \infty,
\end{align}
where 
\begin{align}\label{eq:sigma}
  \sigma_{R^n E C T_C Y^n}
    := \cU_{X} \otimes \id_{Y^nR^nT_C} (\psi_{X^nY^nR^n} \otimes  \Phi_{T_XT_C})
\end{align} 
is the overall pure state after Alice's encoding isometry, and 
\begin{align}\label{eq:omega}
  \omega&_{R^nE\widehat{C}\widehat{T}_C\widehat{Y}^n} \nonumber\\
        &:=  \cD\otimes\id_{ER^n}\bigl( \cU_X\otimes\cE_Y\otimes\id_{T_CR^n}
                                          (\psi_{X^nY^nR^n}\otimes\Phi_{T_XT_C}) \bigr), 
\end{align}
which corresponds to 
$\tr_{\widetilde{E}} \omega_{R^nE\widetilde{E}\widehat{C}\widehat{T}_{C}\widehat{Y}^n}$, 
with $\omega_{R^nE\widetilde{E}\widehat{C}\widehat{T}_{C}\widehat{Y}^n}$ the state
on the right in Fig.~\ref{fig:quantum}.
If there are encodings and decodings as above,
for which Eq.~(\ref{fid}) holds, \emph{i.e.}~the error incurred 
vanishes in the asymptotic limit, then we say that the
corresponding rate pair $(Q_X,Q_Y)$ is {\em{achievable}}. 

Observe that by definition and by the time sharing principle, the set
of achievable rate pairs for a given source is closed, convex and extends
to the above right of the $Q_X-Q_Y$ plane; see Fig.~\ref{Region}. 
Consequently, the achievable region in the plane is entirely described by its
left-lower boundary, the graph of a convex and monotonically non-increasing 
function.} 

\begin{figure}[ht]
  \centering
  \includegraphics[width=0.25\textwidth]{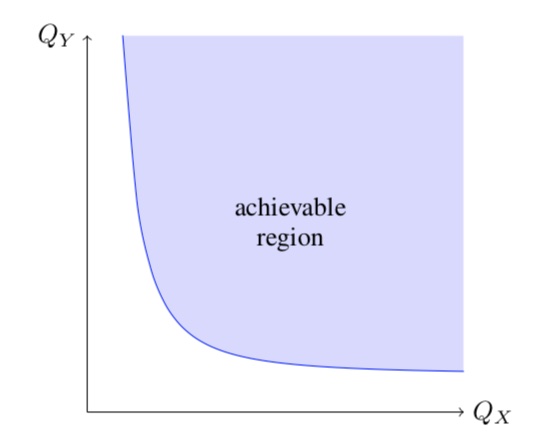}
  \caption{\small Schematic of the achievable rate region for the 
           assisted source coding task. The achievable rate pairs $(Q_X,Q_Y)$ 
           form a closed and convex set extending to the above right.}
  \label{Region}
\end{figure}

\begin{theorem}[Hsieh/Watanabe~{\cite[Thm.~7]{HsiehWatanabe}}]
\label{thmOptRate}
For a given state $\rho_{XY}$ with purification $\psi_{XYR}$, the rate 
pair $(Q_X,Q_Y)$ is achievable if and only if
\begin{equation}\label{eq-thm}
  Q_Y \geq \inf_{\cN^{X \to W}\atop{\frac12 I(W;YR)_\sigma \geq Q_X}} \frac12 I(Y;RV)_\sigma,
\end{equation}
where $\cU_{\cN}$ is the Stinespring isometry of $\cN^{X \to W}$
and $\sigma_{WVYR}:= (\cU_{\cN} \otimes \id_{YR}) \psi_{XYR}$
In other words, the rate pair $(Q_X,Q_Y)$ is achievable if and only if 
\begin{align}
  Q_Y \geq  H(Y) - \frac12 I^q_Y(2Q_X),
\end{align}
where $I^q_Y(R_q)$ is the inverse of $R_q(I^q_Y)$. 
\qed
\end{theorem}

\medskip
The proof of the achievability part of the theorem is illustrated in 
Fig.~\ref{fig:quantum2}. It employs the following two basic protocols
as building blocks: 
the \emph{Quantum Reverse Shannon Theorem (QRST)~\cite{devetak06, abeyesinghe09}, 
aka state splitting}, and \emph{Fully quantum Slepian-Wolf (FQSW)~\cite{abeyesinghe09}, 
aka coherent state merging}.
For completeness, we provide the full proof in the appendix.

\begin{figure}[ht]
  \includegraphics[width=0.5\textwidth]{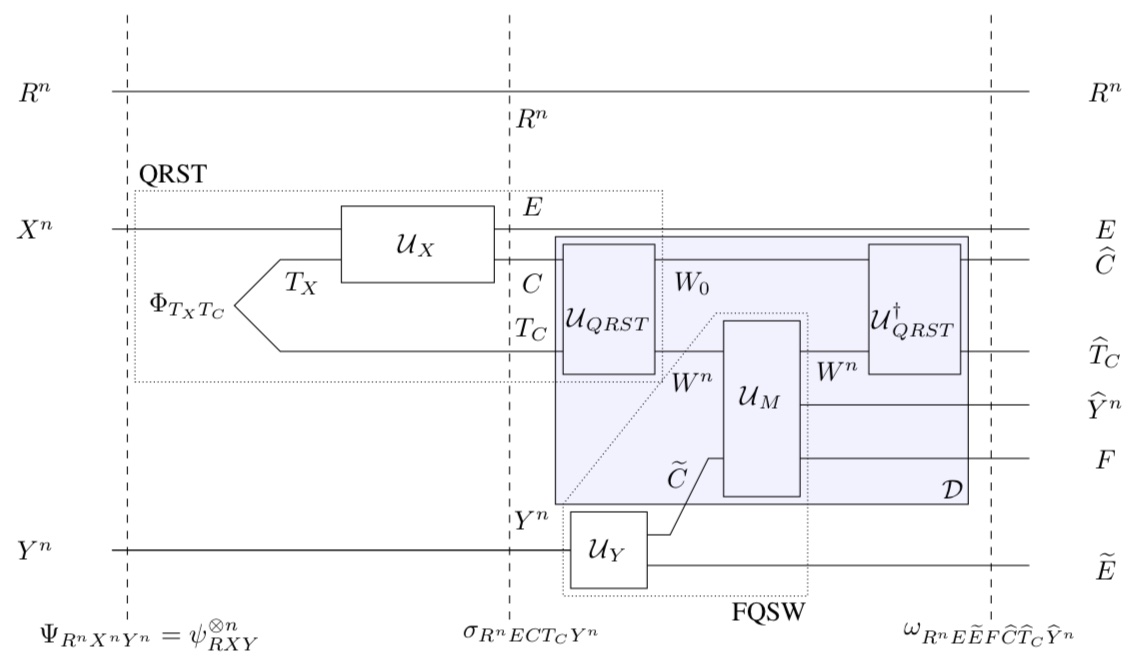}
  \caption{\small A schematic depiction of the protocol used to show achievability 
           in the proof of Theorem~\ref{thmOptRate}. The blue box shows the 
           different actions that compose the decoder ${\mathcal{D}}$. 
           The dotted boxes mark the implementations of the QRST protocol 
           and the FQSW protocol. In particular, $\cU_{QRST}:CT_C\to W^nW_0$ 
           and $\cU_{M}:W^n\widetilde{C}\to W^n \widehat{Y}^n F$ 
           denote the decoding unitary and isometry for QRST and FQSW, respectively.}
  \label{fig:quantum2}
\end{figure}

\section{Privacy funnel}
\label{sec:privacyfunnel}
We can also define a quantum generalization of the so-called 
{\em{privacy funnel function}}, which is closely related to the information bottleneck 
function. The concept of privacy funnel was first introduced in \cite{makhdoumi2014},
where the (classical) privacy funnel function is defined as
\begin{align}\label{pf0}
  G(t) &= \min_{p(w|x)\atop{I(X;W) \geq t}} I(Y;W), \quad \text{for}\ t \geq 0. 
\end{align}
As for the IB function, we can also give a dual function: 
\begin{align}\label{pf1}
  P(a) &= \max_{p(w|x)\atop{I(W;Y) \leq a}} I(X;W), \quad \text{for}\ a \geq 0. 
\end{align}

The underlying motivation can be described as follows: Consider a party who is 
in possession of two correlated sets of data, some public data, $X$, which he 
is willing to disclose, and some private data, $Y$, which she would like to 
keep confidential. A second party (usually called an \textit{analyst}), 
is granted access to all or parts of the public data, and could exploit the 
correlations between $X$ and $Y$ to infer information about the private data. 
The aim of the privacy funnel optimization is to minimize the private information 
leaked, while providing a sufficient amount of public information for the analyst to use.


In analogy to the information bottleneck, we can give a quantum version of the 
privacy funnel by considering the following quantity, cf. Eq.~\eqref{rq-min}:  
\begin{align}\label{pq-min}
  G_q(t) &= \inf_{\cN^{X\to W}\atop{I(YR;W)_\sigma \geq t}} I(Y;W)_\sigma, \quad \text{for}\ t \geq 0.
\end{align}
which again can be equivalently expressed in its dual form
\begin{align}\label{pq-min-dual}
  P_q(a) &= \sup_{\cN^{X\to W}\atop{I(Y;W)_\sigma \leq a}} I(YR;W)_\sigma, \quad \text{for}\ a \geq 0.
\end{align}

\begin{proposition}\label{pq-concavity}
The classical and quantum privacy funnel functions, $G(t)$ and $G_q(t)$, defined 
through Eqs.~\eqref{pf0} and~\eqref{pq-min} are convex, 
\emph{i.e.}~for all $\lambda \in [0,1]$ and $t_0, t_1 \geq 0$,
\begin{align}
  G(\lambda t_0 + (1-\lambda) t_1) &\leq \lambda G(t_0) + (1-\lambda) G(t_1), \\
  G_q(\lambda t_0 + (1-\lambda) t_1) &\leq \lambda G_q(t_0) + (1-\lambda) G_q(t_1).
\end{align}
\end{proposition}
\begin{proof}
The classical case was previously proven in~\cite{du2017}. 
The proof of the quantum version follows immediately from the same 
approach as that of Theorem~\ref{thm1}. 
\end{proof}

\medskip
It would be interesting to find an operational interpretation of $G(t)$ or 
$G_q(t)$. Here, we will not attempt that, but only point out that it probably
will not work along similar lines as we have seen for the information
bottleneck function, \emph{i.e.} as rates in an asymptotic i.i.d.~setting. 
Indeed, in~\cite{du2017} it is shown that the classical privacy funnel 
function is convex and obeys the piecewise linear lower bound
\begin{align}
  G(t) \geq \max\{0,t-H(X|Y)\},
\end{align}
which is 0 up for t between 0 and $H(X|Y)$, and linear with slope
1 for t in the interval from $H(X|Y)$ to $H(X)$. It is also shown
that $G(t)$ in general is different from this lower bound, namely
even in the neighborhood of $t=0$ it is typically positive,
because in~\cite{du2017} it is shown that the derivative at $t=0$ is
typically positive. However, this is not the case for the 
privacy funnel function $G^{(n)}(t)$ of $X^nY^n$, as $n \rightarrow\infty$.

Indeed, we claim that
\begin{align}
  G^{(\infty)}(t) := \inf_n \frac1n G^{(n)}(nt) = \max \{0,t-H(X|Y)\}.
\end{align}
\begin{proof}
Note first that also $G^{(\infty)}(t)$ is convex, and that the
lower bound from~\cite{du2017} still applies, $G^{(\infty)}(t) \geq \max\{0,t-H(X|Y)\}$.
Hence, to show equality, it will be enough to prove that 
$G^{(\infty)}(H(X|Y)) = 0$. This follows from privacy amplification
by random hashing~\cite{bennett1995} of $X^n$ with the eavesdropper's information $Y^n$: 
It is possible to extract $W$ as a deterministic function of $X^n$,
taking values in $\{0,1\}^{nR}$, such that $R$ converges to $H(X|Y)$, and at
the same time $I(W;Y^n)$ goes to $0$.
\end{proof}

\medskip
Since information theoretic interpretations tend to address this
i.i.d.~limit, it seems unlikely that $G(t)$, rather than $G^{(\infty)}(t)$,
can be interpreted in this vein. By analogy, we suspect that
$G_q(t)$ has the same issues of non-additivity, but leave a thorough
discussion of it to another occasion.


\section{Numerics and examples}
\label{Numerics}
In this section we discuss some examples in order to give an intuition for the 
information bottleneck and privacy funnel functions and their properties. For better 
comparison, we choose to normalize the IB function in the following way 
(as is done in \cite{salek2017}): 
\begin{align}\label{nq-IBM}
\oR_q(a) &= \inf_{\cN^{X\to W}\atop{\frac{I(Y;W)_\rho}{I(Y;X)_\rho} \geq a}} \frac{I(X';W)_{\tilde\tau}}{I(X';X)_\tau}, \quad \text{for} \ a\geq 0.
\end{align}
As before, 
\begin{align}\label{taun}
{\tilde{\tau}}_{X'W} &:= ({\rm{id}}_{X'} \otimes \cN^{X\to W})\tau_{X'X},
\end{align}
where $\tau_{X'X}$ is a purification of $\rho_X$, and $\sigma_{WY} := ( \cN^{X\to W} \otimes {\rm{id}}_{Y} )\rho_{XY}$.
This choice of normalization is simply motivated by the data-processing inequality for the mutual information. 

\begin{figure}[t!]
\centering
\begin{overpic}[trim=6cm 10cm 5.7cm 10cm, clip, scale=0.8]{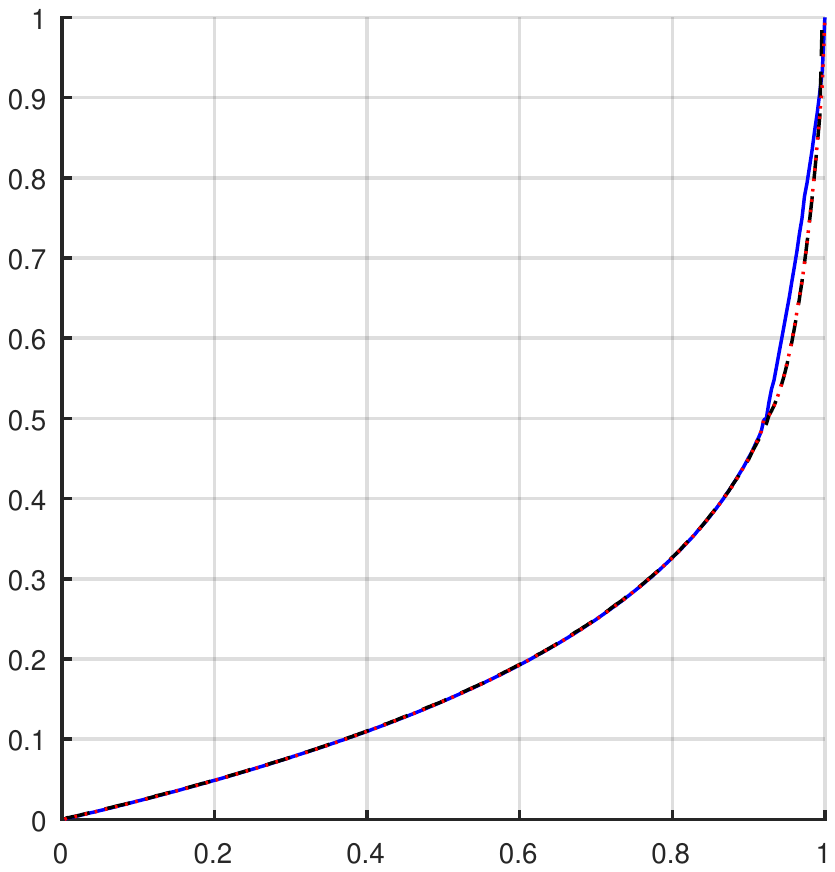}   
\put (-14,213) {$\oR_q(a)$}
\put (108,0) {$a$}
\end{overpic}
\begin{overpic}[trim=0cm 0.15cm 19cm 10.9cm, clip, scale=0.8]{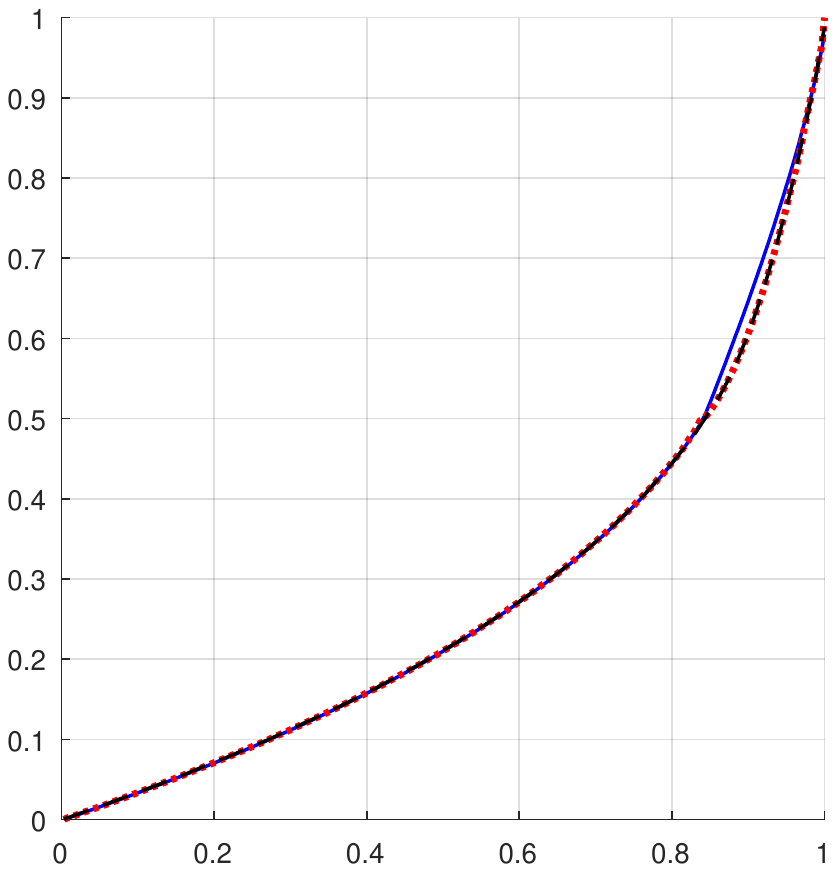}   
\put (-11,211) {$\oR_q(a)$}
\put (109,0) {$a$}
\end{overpic}
\caption{\label{ex3dim} The normalized IB function for the example state $\rho_{XY}^{(3)}$, on the left with $p=0.2$ and on the right with $p=0.4$, each with different allowed dimension for the system $W$. The blue line is $|W|=2$, red is $|W|=3$ and black is $|W|=4$. The red and black line seem to be identical in both cases.}
\end{figure}

For the numerical examples, we use an improved version of the algorithm described in the 
supplemental material of~\cite{salek2017}. We implemented two main features that are not 
present in the original code:
\begin{itemize}
\item In the code provided along with~\cite{salek2017}, not only is $|X|=|Y|=2$ fixed, 
      but also $|W|=2$; we removed the latter restriction. 
\item We adapted the code to also evaluate the privacy funnel function. 
\end{itemize}
Let us first focus on the implications of the first point. \textit{A priori} the size of the system $W$ could be chosen arbitrarily big to aid the optimization. In the classical case it is known that choosing $|W| = |X|+2$ is always sufficient to reach the optimum~\cite{GNT03}. However, in the quantum setting, such a bound is not known and constitutes an important open problem. Now, we can easily give an example where this does in fact play a role. Consider the state
\begin{align}
\rho_{XY}^{(3)} = p \ket{v}\bra{v} + (1-p)\ket{w}\bra{w},
\end{align}
with $\ket{v}=\frac{1}{\sqrt{2}}(\ket{00}+\ket{11})$ and $\ket{w}=\ket{11}$ (this is also example state $\rho_{XY}^{(3)}$ in \cite{salek2017}). 

In Fig.~\ref{ex3dim} the IB function for the above state is given for $p=0.2$ and $p=0.4$. The blue line corresponds to the example given in \cite{salek2017} with $|W|=2$. One can directly see that lower values can be achieved for $|W|>2$, but it also suggests that nothing can be gained from choosing $|W|>3$ (for this particular example). 

Furthermore, this new example suggests that the non-differentiability at $a=0.5$ 
observed in~\cite{salek2017} is rather a result of the $W$ system being too small. 
Interestingly, for a different example (state $\rho_{XY}^{(2)}$ in \cite{salek2017})
we do not find an advantage up to $|W|=16$ and the non-differentiable point remains. 
It might however be possible that $|W|$ simply needs to be chosen to be even larger.

Utilizing the second improvement in the algorithm, we can also plot the privacy funnel (PF) 
function and compare it with the IB function. Again we use a normalized function:
\begin{align}\label{nq-PF}
\oP_q(a) &= \sup_{\cN^{X\to W}\atop{\frac{I(Y;W)_\rho}{I(Y;X)_\rho} \leq a}} \frac{I(X';W)_\tau}{I(X';X)_\tau}, \quad \text{for} \ a\geq 0.
\end{align}
An example for $\rho_{XY}^{(3)}$ with $p=0.4$ can be seen in Fig.~\ref{IBPF}. 

Another simple example is that of a pure state $\rho_{XY}=\psi_{XY}$. Here it is clear that the purification $\tau_{X'X}$ is equivalent to $\psi_{XY}$ up to an isometry on the $Y$ system. Since this isometry commutes with the channel $\cN$, we immediately obtain $I(Y;W)_\rho = I(X';W)_\tau$  and $I(Y;X)_\rho = I(X';X)_\tau$, and therefore $\oR_q(a) = \oP_q(a) =a$. 

 \begin{figure}[t!]
\centering
\begin{overpic}[trim=4.8cm 10.2cm 3cm 10.2cm, clip, scale=0.8]{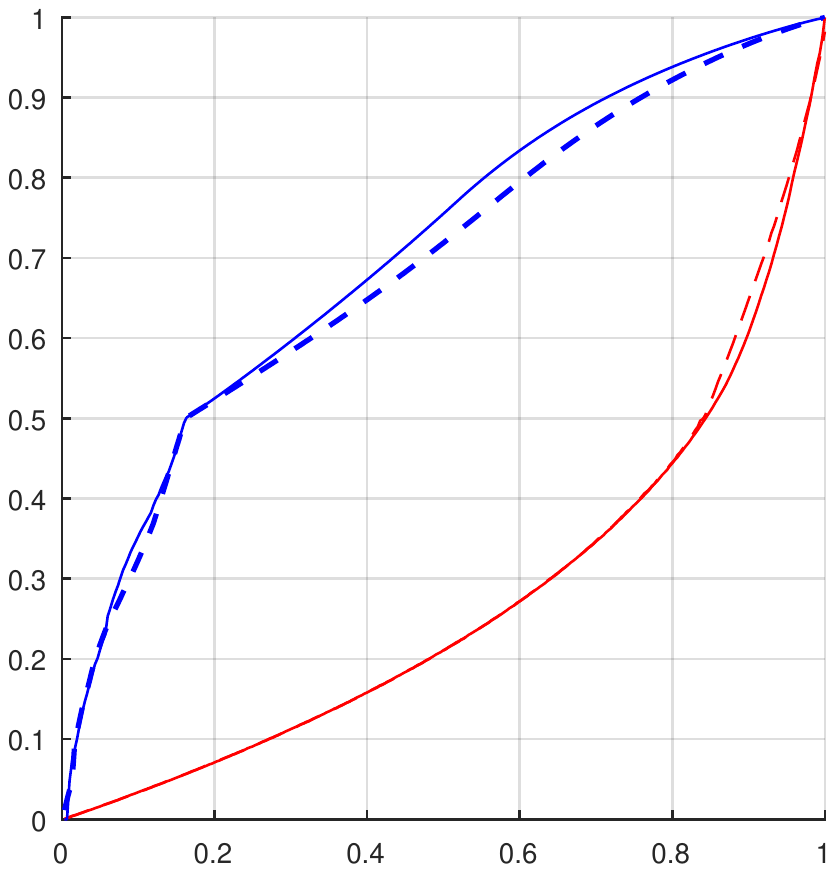}   
\put (15,213) {$\oR_q(a)$ and $\oP_q(a)$}
\put (170,0) {$a$}
\end{overpic}
\caption{\label{IBPF} The normalized IB and privacy funnel functions for $\rho_{XY}^{(3)}$ with $p=0.4$. The IB function is red, the privacy funnel function is blue. The solid lines use $|W|=3$ and the dashed ones $|W|=2$.}
\end{figure}

As a final example, we consider the case where $\rho_{XY}$ is a classical state, {\em{i.e.}}~
\begin{align}
\rho_{XY} = \sum_{x,y} p(x,y) |x\rangle\langle x| \otimes |y\rangle\langle y|
\end{align}
where $\{|x\rangle\}$ and $\{|y\rangle\}$ denote orthonormal bases associated with the
systems $X$ and $Y$, respectively. In~\cite{Grimsmo16} it was shown that for this particular case a classical system $W$ is sufficient for achieving the optimum in the quantum IB function. It follows that the quantum IB function reduces to the classical IB function when the initial state $\rho_{XY}$ is classical. 
This allows us to verify the numerical algorithm we are using to compute the quantum IB function by the following example: 
Consider $X$ and $Y$ to be two binary random variables, with $X$ having a uniform distribution, and $Y$ resulting from the action of a binary symmetric channel, with crossover probability $\delta$, on $X$. This example comes with one particular advantage, that is we can give an analytical expression for the classical IB function. In~\cite{WW1975} it was shown that in this case the following holds: 
\begin{align}
F(a) = h( h^{-1}(a) \star \delta), 
\end{align}
where $F(a)$ is defined in Eq.~\eqref{Eq:condEntrOpt}, $h(x)$ is the binary entropy and $\star$ denotes the binary convolution. This is an important example as it plays a crucial role in the theory of classical and quantum  information combining~\cite{WZ73,hirche18}. 
Now, using the reasoning after Eq.~\eqref{Eq:condEntrOpt}, one can easily get an expression for the classical IB function to which we apply the same normalization as for the quantum function, denoting the result as $\oR(a)$. Plotting values of the classical IB function obtained analytically,
and the values of the quantum IB function obtained numerically, results in Fig.~\ref{classical} and also serves to verify the used numerical algorithm. 

The classical example furthermore exhibits an interesting behavior. Namely, we observe that $\oR_q(1)=\frac{1}{2}$, which turns out to be the same for all classical states $\rho_{XY}$. In~\cite{salek2017} it was suggested that this can be understood in terms of \textit{quantum teleportation}. However, this property should rather be understood as an artifact of the normalization we used when defining $\oR_q(a)$ and the fact that the system $W$ can be chosen to be classical. For the states defined in Eq.~\eqref{taun}, note that $I(X';X)_\tau = 2 H(X)_\tau$ as $\tau$ is a pure state. On the other hand, we can write $I(X';W)_{\tilde\tau} = H(X')_{\tilde\tau} - H(X'|W)_{\tilde\tau}$ and since the conditional entropy is always positive for classical states (but not necessarily for quantum states), we get for this case $I(X';W)_{\tilde\tau} \leq H(X')_{\tilde\tau} = H(X)_\tau$. Applying both to Eq.~\eqref{nq-IBM} we get that for classical states $\rho_{XY}$ we have $\oR_q(a)\leq\frac{1}{2}$.

Note, that, with the same reasoning, $\oR_q(a)\leq\frac{1}{2}$ also holds for quantum states $\rho_{XY}$ when we restrict $W$ to be a classical system (as was previously obtained in~\cite{Grimsmo16}, by considering the unnormalized function). 

\begin{figure}[t!]
\centering
\begin{overpic}[trim=4.35cm 10.2cm 3cm 10.2cm, clip, scale=0.75]{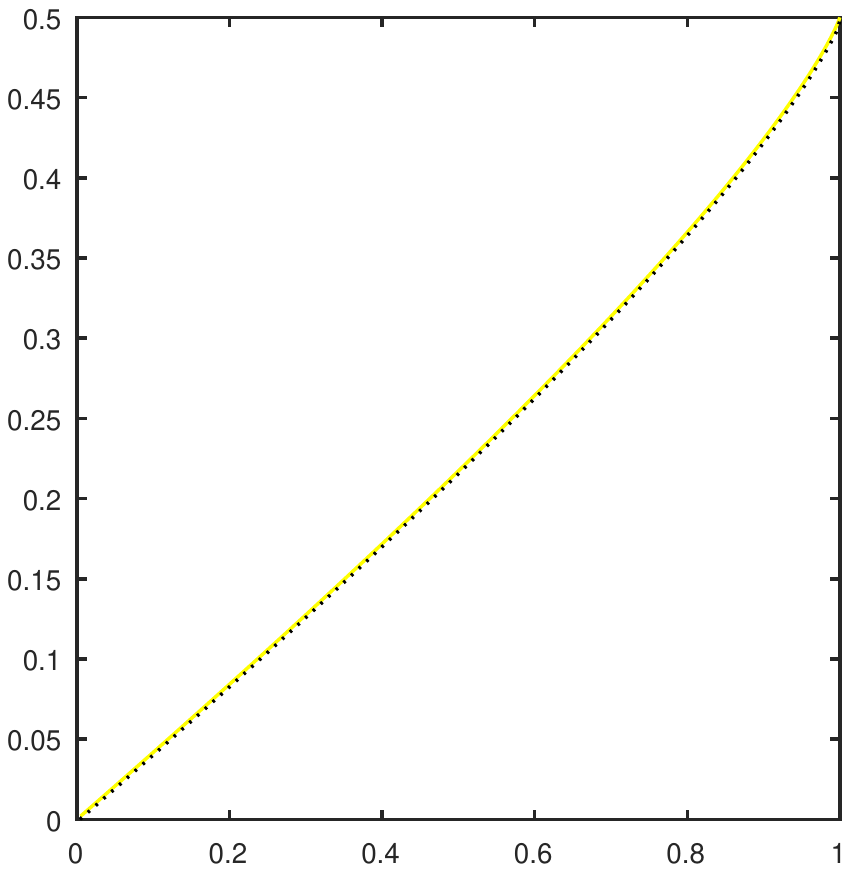}   
\put (15,206) {$\oR_q(a)$ and $\oR(a)$}
\put (170,-2) {$a$}
\end{overpic}
\caption{\label{classical} The normalized classical and quantum IB function for a classical state defined via a binary symmetric channel with $\delta=0.9$. The classical IB function is in yellow and the quantum IB function in black.}
\end{figure}


\section{Conclusions} 
\label{conclusion}
We have demonstrated the convexity of the quantum IB function, completing the 
proof of an operational interpretation for it proposed in~\cite{salek2017}. 
Furthermore we provided a different interpretation coming from source coding 
with side information at the decoder, via prior work by Hsieh and
Watanabe~\cite{HsiehWatanabe}. Along the way we gave an alternative 
formulation of the quantum IB function that might be useful for its 
further investigation. 
%
%

Nevertheless, many open problems remain. 
These include the question whether entanglement is at
all necessary in the source coding task, or if one can remove the 
requirement of exponentially limited amount of entanglement in the converse.

Some other questions are motivated by the properties we know of the classical IB function. 
For example, classically it is an easy consequence of Caratheodory's theorem
that the output dimension of the channel that we optimize over 
can be restricted to $|W| \leq |X|+2$, where $|X|$ is the 
dimension of the input system~\cite{GNT03}. Finding an analogue of this for 
the quantum case would be extremely useful for the evaluation of the quantum
IB function and for its practical application. 

Furthermore, considering the variety of applications of the classical IB function 
it would be interesting to see which of them translate to the quantum setting. 
Finally, the classical IB function is closely related to entropic bounds on 
information combining, and our results might help to better understand their 
quantum generalization~\cite{hirche18}.


\section*{Acknowledgments}
This paper resulted from discussions started at the {\em{Rocky Mountain Summit on
Quantum Information}} in Boulder, in June 2018. The authors thank
the organizers and JILA, University of Colorado, for hospitality.

The authors thank Mark Wilde, Min-Hsiu Hsieh and an anonymous
referee for pointing out that Theorem \ref{thmOptRate} had been found
previously by Hsieh and Watanabe \cite{HsiehWatanabe}.
ND is grateful to Eric Hanson for numerous insightful discussions on the
Information Bottleneck, and to Hao-Chung Cheng for pointing out some typos in 
an earlier version of the paper.

CH and AW acknowledge support from the Spanish MINECO, project
FIS2016-86681-P, with the support of FEDER funds; and
from the Generalitat de Catalunya, CIRIT projects 2014-SGR-966 and 2017-SGR-1127.
CH is in addition supported by FPI scholarship no.~BES-2014-068888.


\bibliographystyle{IEEEtran}


\vspace{0.5cm}

\appendix
\section{Proof of Theorem~\ref{thmOptRate}}

Here we give the complete proof of Theorem~\ref{thmOptRate}; 
it is essentially the proof found in \cite{HsiehWatanabe}, only that
we provide a bit more detail in some places, and that our converse 
proof is based on a general additivity statement, showing also
the additivity of the achievable rate region for a tensor product
of two arbitrary sources.

The following two well-known protocols will be employed in the
following to construct the achievability part of the proof.

\smallskip
\noindent
{(I)} {\em{Quantum Reverse Shannon Theorem (QRST), aka state splitting~\cite{devetak06, abeyesinghe09}:}} 

\noindent
Suppose two parties (say, Alice and Bob) are in different locations but share an 
unlimited amount of entanglement, and let $\cN^{A \to B}$ denote a quantum channel. 
With respect to an
input source $\rho_A^{\otimes n}$ of Alice, the output of the channel $\cN^{\otimes n}$ 
can be simulated at Bob's end, with asymptotically (in $n$) vanishing error, provided 
Alice sends qubits at the following rate to Bob: 
\begin{align}
 \frac{1}{2} I(B;R)_\omega,
\end{align}
where $\omega := {(\cN \otimes \id_R)\psi^\rho_{AR}}$, with $\psi^\rho_{AR}$ being a purification of $\rho_A$.

In more detail, this means the following:
Let ${\mathcal{U}}_{\cN}^{A \to BE}$ denote the Stinespring isometry of the quantum channel, with $E$ denoting the environment. Then, the initial state of the protocol is $\Psi_{A^nR^n} := \left(\psi^\rho_{AR}\right)^{\otimes n} \otimes \Phi_{T_AT_B}$, where by $\Phi_{T_AT_B}$ we denote the entangled state initially shared between Alice and Bob. After execution of the protocol, for $n$ large enough, the purifying reference system $R^n$ remains unchanged, while with asymptotically vanishing error, Alice receives the environment system $E^n$, while Bob receives the system $B^n$ of the state 
$\ket{\psi}_{B^nE^nR^n} := (\ket{\psi}_{BER})^{\otimes n}$, with
$\ket{\psi}_{BER} = {\mathcal{U}}_{\cN}^{A \to BE}|\psi^\rho\rangle_{AR}$.

\medskip
\noindent
{(II)} {\em{Fully Quantum Slepian-Wolf (FQSW), aka coherent state merging~\cite{abeyesinghe09}}:}

Suppose two distant parties (say Bob and Charlie) share the state $\sigma_{B^nC^n}\equiv \sigma_{BC}^{\otimes n}$. Consider its purification $\varphi_{B^nC^nR^n}\equiv \varphi_{BCR}^{\otimes n}$, with $\varphi_{BCR}$ being a purification of $\sigma_{BC}$. Then by implementing the FQSW protocol, Bob can transmit the state of his system $B^n$ (and also the entanglement initially shared between $B^n$ and $R^n$), with asymptotically vanishing error, to Charlie, by sending qubits at a rate
\[
  \frac{1}{2} I(B;R)_\varphi,
\]
to him.

\medskip
\begin{proof-of}[Theorem~\ref{thmOptRate}]
{\em{(Achievability):}} This part of the proof follows directly by applying the previously described QRST and FQSW protocols (see also Fig.~\ref{fig:quantum2}). Let $\rho_X := \tr_{YR} \psi_{XYR}$ and $\psi_{X^nY^nR^n}:= \psi_{XYR}^{\otimes n}$. For $n$ large enough, by using the shared entangled state $\Phi_{T_XT_C}$ and employing the QRST protocol, Alice and Charlie can simulate the output of $n$ independent uses of an encoding map ({\em{i.e.}} a quantum channel) $\cN^{X \to W}$, corresponding to an input 
$\rho_X^{\otimes n}$, with asymptotically vanishing error, provided Alice sends qubits to Charlie, via a system $C$ at a rate $\frac{1}{2}I(W;RY)_\omega$, where $\omega := (\cN^{X \to W} \otimes {\rm{id}}_{YR}) \psi_{XYR}$. 
Once Charlie receives the system $C$ from Alice, the composite system in his possession is 
$\overline{C}:= CT_C$. 

Next, by implementing the FQSW protocol on the tripartite state $\sigma_{R^nECT_CY^n}$, with $R^nE$, $CT_C$ and $Y^n$ being the three systems in the tripartition, Bob can transmit his system $Y^n$ to Charlie, with asymptotically vanishing error, by sending qubits to him at a rate $Q_Y=\frac{1}{2}I(Y;RE)_\sigma= \frac{1}{2}I(Y;RV)_\sigma$, where the equality follows since the system $E$ resulting from applying the FQSW protocol is identical to the purifying system $V$ of the simulated channel $\cN^{X \to W}$.

Minimizing over all possible encoding maps of Alice, yields the expression on the right hand side of \reff{eq-thm}, thus establishing it as an achievable rate of the specified task.
\medskip

{\em{(Converse/Optimality):}}
Consider the most general protocol under which Alice and Bob, by performing local operations and by sending qubits at a rate $Q_X$ and $Q_Y$, respectively, to Charlie, can 
transmit the states $C$, $T_C$ and $Y^n$ to him. 

We assume that there are noiseless quantum channels between Alice and Charlie, and Bob and Charlie. Alice's most general operation may be 
 decomposed into two steps: (i) she locally generates a maximally entangled state $\Phi_{T_XT_C}$ and sends $T_C$ to Charlie; at the
end of this step she has the systems $X^nT_X$, while Charlie has the system $T_C$; (ii) she then applies her encoding map, a CPTP map $\cE^{X^nT_X \to C}$, whose Stinespring isometry we denote as $\cU_{\cE}^{X^nT_X \to CE}$. Let 
\begin{align}\label{sigma}
  \sigma_{R^n E C T_C Y^n}:= \left(\cU_{\cE}^{X^nT_X \to CE}\!\!\otimes \id_{Y^nR^nT_C}\!\right) 
                              \!(\psi_{X^nY^nR^n} \!\otimes\!  \Phi_{T_XT_C}).
\end{align} 
In the above, $X^n = X_1X_2\ldots X_n$,  $Y^n = Y_1Y_2\ldots Y_n$ and  $R^n = R_1R_2\ldots R_n$. 

The rate at which Alice transmits qubits to Charlie
is $Q_X = \frac1n {\log |C|}$. Hence,
\begin{align}
nQ_X & \geq H(C)_\sigma,\nonumber\\
& \geq \frac{1}{2} I(C; Y^nR^nT_C)_\sigma, \nonumber\\
& \geq \frac{1}{2} I(C; Y^nR^n | T_C)_\sigma, \nonumber\\
 &= \frac{1}{2} I(C T_C; Y^nR^n)_\sigma,\nonumber\\
 &\equiv \frac{1}{2} I(\overline{C}; Y^nR^n)_\sigma,\label{qan}
\end{align}
where $\overline{C}= CT_C$; observe that the above reasoning allows us to integrate the steps (i)
 and (ii) of Alice's operation into a single CPTP
 map  $\overline{\cE}:X^n \longrightarrow \overline{C}$, acting as 
 $\overline{\cE}(\rho) = (\cE \otimes \id_{T_C})(\rho_{X^nY^n} \otimes \Phi_{T_XT_C})$, where $\rho_{X^nY^n}= \tr_{R^n} \psi_{X^nY^nR^n}$. Now, the first inequality follows because $H(C) \leq \log |C|$, the second inequality follows from the fact that for a pure state of a tripartite system $ABE$,
\begin{align}
H(A) = \frac{1}{2} I(A;B) + \frac{1}{2} I(A;E) \geq \frac{1}{2} I(A;B), \label{helpIneq}
\end{align}
the third inequality holds because of the chain rule
\begin{align}\label{chain}
I(A;BC) &= I(A;B) + I(A;C|B),
\end{align}
where for any tripartite state $\rho_{ABC}$, $I(A;C|B):= H(\rho_{AB}) + H(\rho_{BC}) - H(\rho_B) - H(\rho_{ABC})$ denotes the conditional mutual information. The equality follows because $T_C$ is uncorrelated with $Y^nR^n$.

Suppose it suffices for Bob to transmit qubits at a rate $Q_Y = \frac1n {\log |{\widetilde{C}}|}$ to Charlie. Let $\cU^{Y^n \to {\widetilde{C}}\widetilde{E}}$ be a unitary that Bob performs, and 
$$\omega_{R^nE{\widetilde{C}}\widetilde{E}} = (\id_{R^nE} \otimes \cU^{Y^n \to {\widetilde{C}}\widetilde{E}}) \sigma_{R^nEY^n}.$$ 
Then $\widetilde{E}$ must be decoupled from $R^nV$ in the asymptotic limit. This is because the fidelity criterion in Eq.~\eqref{fid} ensures that the final state $\omega_{R^nE\widehat{C}\widehat{T}_C\widehat{Y}^n}$ is close to a pure state. 
Indeed, the fidelity bound
$F\left(\sigma_{R^n E C T_C Y^n},\omega_{R^nE\widehat{C}\widehat{T}_C\widehat{Y}^n}\right)^2 \geq 1-\epsilon$
implies, by Uhlmann's theorem,
that there exists a state $\tau_{\widetilde{E}}$ such that
\[
  F\left(\sigma_{R^n E C T_C Y^n}\otimes\tau_{\widetilde{E}},
         \omega_{R^nE\widehat{C}\widehat{T}_C\widehat{Y}^n\widetilde{E}}\right)^2 \geq 1-\epsilon,
\]
which implies in particular that
\[
  F\left(\sigma_{R^n E}\otimes\tau_{\widetilde{E}},
         \omega_{R^n E \widetilde{E}}\right)^2 \geq 1-\epsilon.
\]
Using the well-known relations between fidelity and trace distance, 
this yields
\[
  \frac12 \left\|\sigma_{R^n E}\otimes\tau_{\widetilde{E}}
                 -\omega_{R^n E \widetilde{E}}\right\|_1 \leq \sqrt{\epsilon}.
\]
Hence, by the Alicki-Fannes inequality regarding the continuity of the
quantum conditional entropy~\cite{AlickiFannes,AWcon}, 
for any $\delta>0$ and for large enough $n$, 
\begin{align}\label{Fannes}
  \frac{1}{2} I(\widetilde{E}; R^nE)_\omega &\leq {n{\delta}}.
\end{align}
Namely,
\[\begin{split}
  I(\widetilde{E}; R^nE)_\omega 
     &=    I(\widetilde{E}; R^nE)_\omega-I(\widetilde{E}; R^nE)_{\sigma\otimes\tau} \\
     &\leq \bigl| S(\omega_{\widetilde{E}})-S(\tau_{\widetilde{E}}) \bigr|          \\
     &\phantom{===}
            + \bigl| S(\widetilde{E}|R^n E)_\omega - S(\widetilde{E}|R^n E)_{\sigma\otimes\tau} \bigr| \\
     &\leq 3\sqrt{\epsilon}\log|\widetilde{E}| 
            + 2(1+\sqrt{\epsilon})h\left(\frac{\sqrt{\epsilon}}{1+\sqrt{\epsilon}}\right),
\end{split}\]
where $h(x)=-x\log x - (1-x)\log(1-x)$ is the binary entropy.
To conclude this part of the argument, notice that 
w.l.o.g.~$|\widetilde{E}| \leq |Y|^n |\widetilde{C}| \leq (|Y|2^{Q_Y})^n$,
and so Eq.~(\ref{Fannes}) holds with 
$\delta = (3\log|Y|+3Q_Y)\sqrt{\epsilon}+\frac4n$, which can be made
arbitrarily small for sufficiently large $n$.

Since $Q_Y = \frac1n {\log |{\widetilde{C}}|}$, we have
\begin{align}
\label{qbn}
nQ_Y &\geq H({\widetilde{C}})_\omega,\nonumber\\
     &\geq \frac{1}{2} I({\widetilde{C}}; R^nE\widetilde{E})_\omega,\nonumber\\
     &\geq \frac{1}{2} I({\widetilde{C}}; R^nE|\widetilde{E})_\omega,\nonumber\\
     &=    \frac{1}{2} I({\widetilde{C}}\widetilde{E}; R^nE)_\omega 
            -\frac{1}{2}I(\widetilde{E};R^nE)_\omega\nonumber\\
     &\geq \frac{1}{2} I({Y^n}; R^nE)_\sigma - {n{\delta}},
\end{align}
where the first inequality follows because $H({\widetilde{C}}) \leq \log |{\widetilde{C}}|$, the second inequality follows again from Eq.~(\ref{helpIneq}), 
the third inequality and the equality follow from the chain rule \reff{chain} and the last inequality follows from the fact that the mutual information is invariant under unitaries 
(note that the states of ${\widetilde{C}}\widetilde{E}$ and $Y^n$ are
related by a unitary) and \reff{Fannes}.
\medskip

The final step is now to express the bounds (\ref{qan}) and (\ref{qbn}) on the rates $Q_X$ and $Q_Y$ by single-letter expressions. To this end we define the following set for the pure state $\psi\equiv \psi_{XYR}$:
\begin{equation}\begin{split}
\cT(\psi) &:= \bigl\{ (Q_X,Q_Y) : \exists\,\cU^{X \to WV} \,\text{isometry s.t. } \bigr. \\
          &\phantom{======}
                                   2Q_X \geq I(W;YR)_\sigma,\, 2Q_Y \geq I(Y;RV)_\sigma, \\
          &\phantom{======}
           \bigl. \sigma_{WVYR}:=(\cU^{X\to WV} \otimes \id_{YR}) \psi_{XYR} \bigr\}.
\end{split}\end{equation}
We show below that the set $\cT(\psi)$ satisfies an additivity property: For any two states ${\psi^{(i)}}_{X_iY_iR_i}$,
\begin{align}
\cT(\psi^{(1)}\otimes\psi^{(2)})  = \cT(\psi^{(1)}) + \cT(\psi^{(2)}), \label{Tadd}
\end{align}
where the $+$ on the r.h.s. refers to the Minkowski sum (\emph{i.e.}~element-wise sum) of two sets.
Suppose that $(Q_X,Q_Y) \in {\mathcal{T}}(\psi^{(1)}\otimes \psi^{(2)})$, where $\psi^{(1)}$ (resp.~$\psi^{(2)}$) is a pure state of a tripartite system $X_1Y_1R_1$ (resp.~$X_2Y_2R_2$). Alice possesses the systems $X_1$ and $X_2$, while Bob possesses $Y_1$ and $Y_2$; here $R_1$ and $R_2$ denote inaccessible, purifying reference systems. The final composite pure state, resulting from the action of a Stinespring isometry $\cU^{X_1X_2 \to WV}$, is then given by 
\begin{align}\label{wsigma}
\wsigma_{WVY_1Y_2R_1R_2} := \left(\cU^{X_1X_2 \to WV} \otimes \id_{Y_1Y_2 R_1R_2}\right) (\psi^{(1)}\otimes\psi^{(2)}).
\end{align}
The $\supseteq$ direction follows directly from the fact that $(i)$ the set of all isometries $\cU^{X_1X_2 \to WV}$ clearly also includes all those of the form $\cU^{X_1 \to W_1V_1}\otimes\cU^{X_2\to W_2V_2}$, with $\cU^{X_1 \to W_1V_1}$ and $\cU^{X_2\to W_2V_2}$ being isometries arising in the definitions of the sets ${\mathcal{T}}(\psi^{(1)})$ and ${\mathcal{T}}(\psi^{(2)})$, and $(ii)$ from the additivity of the mutual information. We will therefore concentrate on the $\subseteq$ direction. 

By assumption, and using the chain rule Eq.~(\ref{chain}), we obtain
\begin{align}
2Q_{X_1X_2} & \geq  I(W; Y_1Y_2R_1R_2)_{\wsigma},\nonumber\\[2pt]
&=  I(W ; Y_1R_1|Y_2R_2)_{\wsigma}  + I(W ; Y_2R_2)_{\wsigma}       \nonumber\\[2pt]
&= I( Y_1R_1 ;W Y_2R_2)_{\wsigma}  +  I(W ; Y_2R_2)_{\wsigma}       \nonumber\\[2pt]
&=: I(Y_1R_1; W_1 )_{\wsigma} + I( Y_2R_2; W_2)_{\wsigma}, \label{Eqn:boundQX}
\end{align}
with $W_1:=W Y_2R_2$ and $W_2:=W$.

Similarly, by using the chain rule (\ref{chain}) twice, the fact that the states of the systems $Y_1$ and $Y_2$ are uncorrelated (and hence
$I(Y_1;Y_2)_\sigma=0$) and the data-processing inequality (with respect to partial trace) we obtain
\begin{align}
2Q_{Y_1Y_2} &\geq I(Y_1Y_2; R_1R_2V)_{\wsigma},\nonumber\\[2pt]
&= I( Y_1; R_1R_2V)_{\wsigma} + I(Y_2; R_1R_2V|Y_1)_{\wsigma} \nonumber\\[2pt]
&=I( Y_1; R_1R_2V)_{\wsigma} + I(Y_2; R_1R_2VY_1)_{\wsigma} - I(Y_1;Y_2)_\sigma \nonumber\\[2pt]
&\geq I( Y_1; R_1V)_{\wsigma} + I(Y_2; R_1R_2VY_1)_{\wsigma}  \nonumber\\[2pt]
&=: I( Y_1; R_1V_1)_{\wsigma} + I(Y_2; R_2V_2)_{\wsigma},   \label{Eqn:boundQY}
\end{align}
with $V_1:=V$ and $V_2:=VR_1Y_1$.

The pure state $\wsigma$ of Eq.~(\ref{wsigma}) of the composite system $WVY_1Y_2R_1R_2$ results from the action of the isometry $\cU^{X_1X_2 \to WV}$ on $\psi^{(1)}\otimes\psi^{(2)}$. However, by the above definitions of the systems $W_1$, $W_2$, $V_1$ and $V_2$, it follows that one can construct two isometries $\mathcal{U}^{(i)}:X_i \rightarrow W_iV_i$, for $i=1,2$, which when acting solely on the pure state $\psi^{(1)}$ and $\psi^{(2)}$ respectively, yields pure states of this same composite system $WVY_1Y_2R_1R_2$. Let these resulting pure states be denoted as $\sigma^{(1)}$ and $\sigma^{(2)}$ respectively:
\begin{align*}
\sigma^{(1)}\equiv \sigma^{(1)}_{W_1V_1Y_1R_1} &=  (\cU^{(1)}\otimes \id_{Y_1R_1}) \psi_{X_1Y_1R_1},\\
\sigma^{(2)}\equiv \sigma^{(2)}_{W_2V_2Y_2R_2} &=  (\cU^{(2)}\otimes \id_{Y_2R_2}) \psi_{X_2Y_2R_2}.
\end{align*}
Then the bounds (\ref{Eqn:boundQX}) and (\ref{Eqn:boundQY}) can be rewritten as
\begin{align*}
2Q_{X_1X_2} &\geq I( Y_1R_1; W_1)_{\sigma^{(1)}} + I(Y_2R_2;W_2)_{\sigma^{(2)}},\\
2Q_{Y_1Y_2} &\geq I( Y_1; R_1V_1)_{\sigma^{(1)}} + I(Y_2; R_2V_2)_{\sigma^{(2)}}.
\end{align*} 

\vfill\pagebreak
However, by definition of the isometries $\cU^{(i)}$ for $i=1,2$, it 
follows that $I( Y_iR_i; W_i)_{\sigma^{(i)}}$ and $I( Y_i; R_iV_i)_{\sigma^{(i)}}$ 
are respectively valid lower bounds on $2Q_{X_i}$ and $2Q_{Y_i}$ for pairs 
$(Q_{X_i},Q_{Y_i})$ occurring in the sets $\cT(\psi^{(i)})$.

Hence, we have
\[
  \cT(\psi^{(1)}\otimes\psi^{(2)})  \subseteq \cT(\psi^{(1)}) + \cT(\psi^{(2)}),
\]
and we conclude the additivity property Eq.~\eqref{Tadd}. 
As an immediate implication, we get by induction, choosing 
$\psi^{(1)}=\psi$ and $\psi^{(2)}=\psi^{\otimes (n-i)}$ 
for $i\in\{1,\dots,n-1\}$, that
\begin{align}
  \cT(\psi^{\otimes n}) = n \cT(\psi).
\end{align}

It follows, returning to the bounds~\eqref{qan} and~\eqref{qbn}, 
that there is an isometry ${\mathcal U}:X \longrightarrow VW$ such that
 \begin{align}
   Q_X &\geq \frac12 I(W;YR)_\sigma,\\
   Q_Y &\geq \frac12 I(Y;RV)_\sigma-{{\delta}},
 \end{align}
where $\sigma_{WVYR}$ is as in the statement of the theorem.
Since $\delta$ becomes arbitrarily small for sufficiently large $n$, 
we obtain the desired bounds.
\end{proof-of}

\medskip
\begin{remark}
In the above proof we have used Eq.~\eqref{Tadd} only to show the
single-letterization of the rate region in Theorem \ref{thmOptRate}.
Since the result is that $\cT(\psi)$ is that rate region, the additivity
relation~\eqref{Tadd} shows that the rate region of a product of
two independent sources is the Minkowski sum of the individual rate
regions. 
\end{remark}

\end{document}